\documentclass{llncs}
\usepackage{etex}
\usepackage{macro}

\title{Incentive Stackelberg 
	Mean-payoff Games}

\begin{document}
	\author{
		Anshul Gupta\inst{1} \and 
		M. S. Krishna Deepak \inst{2} \and
		Bharath Kumar Padarthi\inst{2} \and \\
		Sven Schewe\inst{1} \and Ashutosh Trivedi\inst{2}} 
	
	\institute{
	  University of Liverpool, UK
          \and 
          Indian Institute of Technology Bombay, India
	}
	\maketitle
	
\begin{abstract}
	We introduce and study \emph{incentive equilibria} for
	multi-player mean-payoff games. 
	Incentive equilibria generalise well-studied solution concepts such as Nash
	equilibria and leader equilibria (also known as Stackelberg equilibria). 
	Recall that a strategy profile is a Nash equilibrium if no player can
	improve his payoff by changing his strategy unilaterally. 
	In the setting of incentive and leader equilibria, there is a distinguished
	player---called the leader---who can assign strategies to all other players,
	referred to as her followers.  
	A strategy profile is a leader strategy profile if no player, except for the leader, can
	improve his payoff by changing his strategy unilaterally, and a leader equilibrium is a leader strategy profile with a maximal return for the leader.  
	In the proposed case of incentive equilibria, the leader can additionally
	influence the behaviour of her followers by transferring parts of her payoff to
	her followers. 
	The ability to incentivise her followers provides the leader with more freedom
	in selecting strategy profiles, and we show that this can indeed improve the
	leader's payoff in such games. 
	The key fundamental result of the paper is the existence of incentive equilibria
	in mean-payoff games.   
	We further show that the decision problem related to constructing incentive
	equilibria is NP-complete. 
	On a positive note, we show that, when the number of players is fixed, the
	complexity of the problem falls in the same class as two-player mean-payoff
	games.  
	We also present an implementation of the proposed algorithms, and discuss
	experimental results that demonstrate the feasibility of the analysis of medium
	sized games.   
\end{abstract}

\section{Introduction}
The classical mean-payoff games~\cite{Zwick+Paterson/96/payoff,DBLP:conf/lics/ChatterjeeHJ05} are
two-player zero-sum games that are played on weighted finite di-graphs, where two
players---Max and Min---take turn to move a token along the edges of the graph
to jointly construct an infinite play.
The objectives of the players Max and Min are to respectively maximise
and minimise the limit average reward associated with the play.  
Mean-payoff games enjoy a special status in verification, since $\mu$-calculus
model checking and parity games can be reduced in polynomial-time to solving
mean-payoff games.   
Mean-payoff objectives can also be considered as quantitative
extensions~\cite{Henzinger13,6940380} of 
classical B\"uchi objectives, where we are interested in the limit-average share of
occurrences of accepting states rather than merely in whether or not infinitely
many accepting states occur.  
For a broader discussion on quantitative verification, in general, and the transition
from the classical qualitative to the modern quantitative interpretation of
deterministic B\"uchi automata, we refer the reader to Henzinger's survey
on quantiative reactive modelling and verification \cite{Henzinger13}. 
We focus on multi-player extension of mean-payoff games where a
finite number of players control various vertices and 
move a token along the edges to collectively produce an infinite run. 
There is a player-specific reward function that, for every edge of the graph,
gives an immediate reward to each player. 
The payoff to a player associated with a play is the limit average of
the rewards in the individual moves.
The most natural question related to the multi-player game setting is to find an
optimal `stable' \emph{strategy profile} (a set of strategies, one for each
player). 
Broadly speaking, a strategy profile is stable, if no player has an incentive to
deviate from it. 
Nash equilibria~\cite{Nash01011950} and leader 
equilibria~\cite{von1934marktform,6940380} (also known as \emph{Stackelberg
	equilibria}) are the most common notions of stable strategy profiles for
multi-player games.
A strategy profile is called a Nash equilibrium if no player
can improve his payoff by unilaterally changing his strategy. 
In a setting where we have a distinguished player (called the leader) who is
able to suggest a strategy profile to other players (called followers), a
strategy profile is stable if no follower can improve his payoff by
unilaterally deviating from the profile. 
A leader equilibrium is a stable strategy profile that maximises the reward
for the leader.

\begin{wrapfigure}[7]{r}{8cm} 
	\vspace{-8mm}
	\begin{tikzpicture}[->,>=stealth',shorten >=1pt]
	\tikzstyle{vertex}=[circle,fill=black!10,minimum size=17pt,inner sep=0pt,font=\sffamily\small\bfseries]
	\tikzstyle{rvertex}=[circle,fill=red!10,minimum size=17pt,inner sep=0pt,font=\sffamily\small\bfseries]
	\tikzstyle{gvertex}=[circle,fill=green!10,minimum size=17pt,inner sep=0pt,font=\sffamily\small\bfseries]
	\node (6) at (0.25,1) {};
	\node (1) at (1,1) [gvertex,draw]{$1$} ; 
	\node (2) at (3.5,1) [rvertex,draw] {$2$};
	\node (4) at (1,0) [vertex,draw] {$4$};
	\node (3) at (6,1) [vertex,draw] {$3$};
	\node (5) at (3.5,0) [vertex,draw] {$5$};
	
	\path[every node/.style={font=\sffamily\small}]
	(6) edge [right] node[] {} (1)
	(1) edge [right] node[] {} (2)       
	edge [below] node[below] {} (4)
	(2) edge [below] node[below] {} (5)
	edge [right] node[] {} (3)
	(3) edge [loop right] node[] {$(0,9,-9)$} (3)        
	(4) edge [loop right] node[] {$(1,0,-1)$} (4)
	(5) edge [loop right] node[] {$(1,1,-2)$} (5);
	
	\end{tikzpicture}
	\caption{Incentive equilibrium beats leader equilibrium beats Nash equilibrium.}
	\label{fig:incentiveisbetter}
\end{wrapfigure}


In this paper, we introduce and study a novel notion of stable strategy profiles for
multi-player mean-payoff games that we call incentive Stackelberg equilibria (or
incentive equilibria for short). 
In this setting, the leader has more powerful strategies, where she not only puts
forward strategies that describe how the players move, but also gives non-negative
incentives to the followers for compliance. 
These incentives are then added to the overall rewards the respective follower would receive
in each move of the play, and deduced from the overall reward of the leader. 
Like for leader equilibria, a strategy profile is stable if no \emph{follower}
has an incentive to deviate. 
An \emph{incentive equilibrium} is a stable strategy profile with maximal reward
for the leader. 

Using incentive equilibria has various natural justifications. 
The techniques we discussed here can be applied where distributed development of a system is
considered. That is, when several rational components interact among themselves along
with a rational controller and they try to optimise their individual
objectives and specifications. 
Our techniques can be applied to maximise utility of a
central controller while also complying with individual component
specifications. 
Transferring utilities is also quite natural where the payoffs on the edges
directly translate to the gains incurred by individual components. 
These techniques can also be used to maximise social optima where rational
controller follow the objective of maximising joint utility.  

We now discuss two simple examples that exemplify the role that incentives can
play to achieve good stable solutions of multi-player mean-payoff games.


\begin{example}
	Consider the multi-player mean-payoff game shown in Figure~\ref{fig:incentiveisbetter}.
	Here we have three players: Player~1, Player~2 (the leader), and Player~3. 
	The vertex labelled~$1$ is controlled by Player~1, while the vertex labelled $2$
	is controlled by Player~$2$. 
	All other vertices are controlled by Player~$3$. 
	We further annotate the rewards of various players on the edges of the graph
	by giving a triple, where the reward of the players~$1$, $2$, and~$3$ are
	shown in that order. 
	We omit the labels when the rewards of all players are~$0$.
	An incentive equilibrium would be given by (a strategy profile leading to) the play
	$\seq{1, 2, 3^\omega}$, where the leader pays
	an incentive of $1$ to Player $1$ for each step and $0$ to Player~$3$.
	By doing this, she secures a payoff of $8$ for herself. 
	The reward for the players~$1$ and~$3$ in this incentive equilibrium are each $1$ and $-9$, respectively.
	A leader equilibrium would result in the play $\seq{1, 2, 5^\omega}$ (with rewards) of $1$ for Player 1 and the leader and $-2$ for Player~3: 
	when the leader cannot pay any incentive to Player $1$, then the move from Vertex $2$ to Vertex $3$ will not be part of a stable strategy. 
	The only Nash equilibrium in this game would result in the play $\seq{1, 4^\omega}$ with the rewards of $1$ for Player 1, $0$ for the leader, and $-1$ for Player~3.  
	This example therefore shows how the leader can benefit from her additional choices in leader and incentive equilibria.
	\qed
\end{example}

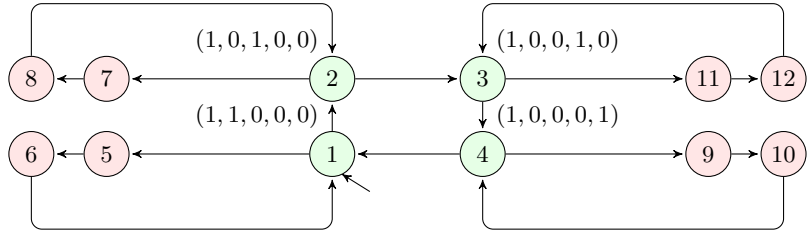
\begin{figure*}[t] 
	\begin{center}
		\begin{tikzpicture}[->,>=stealth',shorten >=1pt,rounded corners]
		\tikzstyle{rvertex}=[draw, circle,fill=red!10,minimum size=17pt,inner sep=0pt,font=\sffamily\small\bfseries]
		\tikzstyle{gvertex}=[draw, circle,fill=green!10,minimum size=17pt,inner sep=0pt,font=\sffamily\small\bfseries]
		
		\node[gvertex] (1) at (4, 1) {$1$};
		\node[gvertex] (2) at (4, 2) {$2$};
		\node[gvertex] (4) at (6, 1) {$4$};
		\node[gvertex] (3) at (6,2) {$3$};
		
		\node[rvertex] (5) at (1, 1) {$5$};
		\node[rvertex] (6) at (0, 1) {$6$};
		
		\node[rvertex] (7) at (1, 2) {$7$};
		\node[rvertex] (8) at (0, 2) {$8$};
		
		\node[rvertex] (9) at (9, 1) {$9$};
		\node[rvertex] (10) at (10, 1) {$10$};
		
		\node[rvertex] (11) at (9, 2) {$11$};
		\node[rvertex] (12) at (10, 2) {$12$};

		\draw[->] (1) -- (2);
		\draw[->] (2) -- (3);
		\draw[->] (3) -- (4);
		\draw[->] (4) -- (1);
		\draw[->] (1) -- (5);
		\draw[->] (5) -- (6);
		\draw[->] (6) -- (0, 0) -- (4, 0) -- (1);;
		\draw[->] (2) -- (7);
		\draw[->] (7) -- (8);
		\draw[->] (8) -- (0, 3) -- (4, 3) -- (2);;
		\draw[->] (3) -- (11);
		\draw[->] (11) -- (12);
		\draw[->] (12)  -- (10, 3) -- (6, 3) -- (3);;
		\draw[->] (4) -- (9);
		\draw[->] (9) -- (10);
		\draw[->] (10)  -- (10, 0) -- (6, 0) -- (4);;
		
		\draw[->] (4.5, 0.5) -- (4.1, 0.75);
		
		\node at (3, 1.5) {$(1, 1, 0, 0, 0)$};
		\node at (3, 2.5) {$(1, 0, 1, 0, 0)$};
		\node at (7, 1.5) {$(1, 0, 0, 0, 1)$};
		\node at (7, 2.5) {$(1, 0, 0, 1, 0)$};
		
		\end{tikzpicture}
	\end{center}
	\caption{Incentive equilibrium gives much better system utilisation.}
	\label{fig:intro-token}
\end{figure*}

\begin{example}
	Consider the multi-player mean-payoff game shown in the
	Figure~\ref{fig:intro-token} with five players---Player~1 (or: leader) and 
	Player~2 to 5 (followers). 
	For $i \in \{2, 3, 4, 5\}$, Player~$i$ controls the vertex labelled~$i$ in the game and
	gets a reward of $1$ whenever token is at vertex~$i$. (To keep the rewards on the edges, one could encode this by giving this reward whenever vertex $i$ is entered.)
	Player $1$ gets a reward of $1$ in all of these vertices.
	The payoff of all other players is $0$ in all other cases.  
	Notice that the only play defined by Nash or leader equilibria in this example is 
	$\seq{(1, 5, 6)^\omega}$, which provides a payoff of $\frac{1}{3}$ to Player $1$
	and Player~$2$, and a payoff of~$0$ to all other players. 
	For incentive equilibria, however, the leader can give
	an incentive of $\frac{1}{12}$ to all followers when they follow the play
	$\seq{(1, 2, 3, 4)^\omega}$. It is easy to see that such a strategy profile is incentive stable.
	The leader will then receive a payoff
	of $\frac{2}{3}$, i.e., her payoff from the cycle, $1$, minus the incentives given to
	the other players, $4\cdot \frac{1}{12}$.
	All other players receive a payoff of~$\frac{1}{3}$, consisting of the payoff from the cycle, $\frac{1}{4}$, plus the incentive they receive from the leader,
	$\frac{1}{12}$.   
	Notice that this payoff is not only better from the leader's point-of-view,
	the other players are also better off in this equilibrium.
	\qed
\end{example}

In both examples, we saw that the incentive equilibria are strictly better than
Nash and leader equilibria.  
It is not a coincidence---note that leader reward from any Nash equilibrium cannot be greater than her
reward from any leader equilibrium, as in the case of leader strategy profiles,
leader can select from a wider range of strategy profiles.  
Thus, if compared to a Nash equilibrium, a leader equilibrium can only be superior
w.r.t.\ the leader reward. 
Similarly, a leader equilibrium cannot beat an incentive equilibrium, as here also,
leader can select from a wider range of strategy profiles (`leader stable'
strategy profiles can be viewed as an `incentive stable' strategy profiles with
$0$ incentives). 
It again implies that leader reward from any leader equilibrium
cannot be greater than her reward from any incentive equilibrium. 

\emph{\textbf{Related Work.}}
Ummels and Wojtczak~\cite{DBLP:conf/concur/UmmelsW11,DBLP:conf/fossacs/Ummels08} considered Nash equilibria
for mean-payoff games and showed that the decision problem of finding a Nash
equilibria is NP-complete for pure (not allowing randomisation) strategy
profiles, while the problem is undecidable for arbitrary randomised strategies. 
Gupta and Schewe~\cite{6940380} have extended these results to leader equilibria.
The undecidability result of~\cite{DBLP:conf/concur/UmmelsW11} for
Nash-equilibria in arbitrary randomised strategies can be easily extended to
leader equilibria.
For this reason, we focus on non-randomised strategies throughout
this paper. 

Leader equilibria were introduced by von Stackelberg~\cite{von1934marktform} and
were further studied in~\cite{friedman1977oligopoly}.   
The strategy profiles we study here are inspired from
\cite{Friedman71Noncooperative} and are studied in detail for infinite games in
\cite{6940380}.  
Incentive equilibria have recently been introduced for bi-matrix
games~\cite{GS/14/bimatrix}, but have, to the best of our knowledge, not been
used in infinite games. 
Two-player mean-payoff games were first studied in \cite{positional/1979} 
and were shown to be positionally determined. 
They can be solved in pseudo-polynomial time
\cite{Zwick+Paterson/96/payoff,DBLP:journals/fmsd/BrimCDGR11}, smoothed
polynomial time \cite{DBLP:conf/icalp/BorosEFGMM11}, PPAD \cite{DBLP:journals/siamcomp/EtessamiY10}
and randomised subexponential~\cite{BjorklundVorobyov/07/subexp} time. 
Their decision problem is also known to be in
UP$\cap$co-UP~\cite{Jurdzinski/98/UP,Zwick+Paterson/96/payoff}.   

\emph{\textbf{Contributions.}}
The key contribution of the paper is the concept of incentive equilibria to
system analysis in general and to multi-player mean-payoff games in particular.  
We show that the complexity of finding incentive equilibria is same as that for
finding leader equilibria~\cite{6940380} for multi-player mean-payoff
games: it is NP-complete in general, but, for a fixed number of players, it is in the same complexity class as solving two-player mean-payoff games (2MPGs).  
In other words, solving two-player mean-payoff games is the most expensive step
involved. 
We have implemented an efficient version of the optimal strategy improvement
algorithm from~\cite{Schewe/08/improvement} as a backbone, and equipped it with
a logarithmic search to expand it from the qualitative evaluation (finding mean
partitions) of mean-payoff games to their quantitative evaluation. 
We construct incentive equilibria by implementing
a constraint system that gives necessary and sufficient conditions for a
strategy profile to be (1) stable and (2) provide optimal leader return among them.
The evaluation of the constraint system involves evaluating a bounded number of
calls to the linear programming solver.      

The contribution of the paper is therefore two-fold---first to conceptualise
incentive equilibria in multi-player mean-payoff games, and 
second to present a tool deriving optimal return for the leader by evaluating
a number of multi-player games. 

\emph{\textbf{Organisation.}}
We begin the technical presentation by formally introducing incentive equilibria for
multi-player mean-payoff games.
In Section~\ref{sec:exist} we present details related to existence and
construction of incentive equilibrium and declare the complexity of finding
incentive equilibrium. 
In Section~\ref{sec:tooldetails} we discuss details of our
implementation of constructing incentive equilibrium before concluding in Section~\ref{sec:dis}.

\section{Incentive equilibrium}
We introduce the concept of incentive equilibria for multi-player mean-payoff
games. 
These games are played among multiple players on a multi-weighted
finite directed graph arena where a distinguished player, called the leader, is able to put forward a strategy profile (a strategy each for \emph{all} players).
She will follow the strategy she assigned for herself, while all other players, called her followers, will comply with the strategy she suggested, unless they benefit from unilateral deviation.
The leader is further allowed to incentivise the behaviour of her followers by
sharing her payoff with them, in order to make compliance with the strategy she has put forward sufficiently attractive.
This, in turn, may improve the leader’s payoff. 
Before we define incentive equilibra, let us recall a few key
definitions. 
\begin{definition}[Multi-player Mean-Payoff Game Arena]
	A  multi-player mean-payoff game (MMPG) arena $\mathcal{G}$ is a tuple $(P,
	V, (V_p)_{p \in P}, v_0, E, (r_p)_{p \in P})$ where
	\begin{itemize}
		\item 
		$P$ is a finite set of players with a distinguished leader player $l
		\in P$,
		\item
		$V$ is a finite set of vertices with a distinguished initial vertex $v_0 \in
		V$, 
		\item
		$(V_p)_{p \in P}$ is a partition of $V$ characterising vertices controlled
		by players, 
		\item
		$E \subseteq V \times V$ is a set edges s.t. for all $v \in V$
		there is $v' \in V$ with $(v,v'){\in} E$, 
		\item
		$(r_p)_{p \in P}$ is a family of reward functions $r_p: E \rightarrow
		\mathbb Q$, that for each player $p \in P$, assigns reward for player $p$
		associated with that edge. 
	\end{itemize}
\end{definition}
A finite play $\pi = \seq{v_0, v_1, \ldots, v_n}$ of the game $\mathcal{G}$ is
a sequence of vertices such that $v_0$ is the initial vertex, and for every $0
\leq i < n$, we have, $(v_i, v_{i+1}) \in E$. 
An infinite play is defined in an analogous manner. 
A multi-player mean-payoff game is played on a game arena $\mathcal{G}$
among various players by moving a token along the edges of the arena. 
The game begins by placing a token on the initial vertex.
Each time the token is on the vertex controlled by a player $p \in P$, the
player $p$ chooses an outgoing edge and moves the token along this edge. 
The game continues in this fashion forever, and the players thus construct an
infinite play of the game. 
The (raw) payoff $r_p(\pi)$ of a player $p \in P$ associated with a play $\pi =
\seq{v_0, v_1, \ldots}$ is the limit average reward of the path, given as 
$r_p(\pi)  
\rmdef \liminf_{n\rightarrow \infty} \frac{1}{n}\sum_{i=0}^{n-1}r_p
\big((v_i,v_{i+1})\big)$.
We refer to this value as the raw payoff of the player $p$ to distinguish it
from the payoff for the player that also includes the incentive given to the
player by the leader. 

A strategy of a player is a recipe for the player to choose the successor vertex.
It is given as a function $\sigma_p: V^*V_p \rightarrow V$ such that
$\sigma_p(\pi)$ is defined for a finite play $\langle v_0, \ldots, v_n \rangle$ when $v_n
\in V_p$ and it is such that $(v_n, \sigma_p(\pi)) \in E$.
A family of strategies $\osigma = (\sigma_p)_{p \in P}$ is called a strategy
profile.
Given a strategy profile $\osigma$, we write $\osigma(p)$ for the strategy of
player $p \in P$ in $\osigma$. 
A strategy profile $\osigma$ defines a unique play $\pi_\osigma$, and therefore a
raw payoff $r_p(\osigma) = r_p(\pi_{\osigma})$ for each player $p$.
We write $\Sigma_p^\Gg$ for the set of strategy of player $p \in P$ and $\Pi^\Gg$ for
the set of strategy profiles in a game arena $\Gg$. When the game arena is clear
from the context, we omit it from the superscript. 

For a strategy profile $\osigma$, a player $p \in P$, and a strategy $\sigma'$
of $p$, we write $\osigma_{p,\sigma'}$ for the strategy profile
$\osigma'$ such that $\osigma'(p) = \sigma'$ and $\osigma'(p') = \osigma(p')$ for
all $p' \in P \setminus \set{p}$.
We are now in a position to formally define Nash and leader (aka Stackelberg)
equilibra.

\begin{definition}[Nash Equilibria]
	A strategy profile $\osigma$ is 
	a Nash equilibrium if no player would gain
	from unilateral deviation, i.e., for all players $p \in P$ we have 
	$r_p(\osigma) \geq r_p(\osigma_{p, \sigma'}) \text{ for all } \sigma' \in
	\Sigma_p$. 
\end{definition}

\begin{definition}[Leader Equilibrium]
	A strategy profile $\osigma$ is 
	a leader stratey profile if no player, except for the leader, would gain from
	unilateral deviation, i.e., for all $p \in P{\setminus}\set{l}$ we have  
	$r_p(\osigma) \geq r_p(\osigma_{p, \sigma'})$ for all $\sigma' \in
	\Sigma_p$.    
	A leader equilibrium is a maximal (w.r.t.\ leader's raw payoff) leader strategy
	profile. 
\end{definition}

\paragraph{Incentive Equilibrium.} We next define an \emph{incentive strategy
	profile} as a strategy profile which  satisfies the stability requirements of
the leader equilibria and allows the leader to give incentives to the followers. 
We refer to an optimal strategy profile in this class of strategy profiles that
provides maximal reward to the leader as an \emph{incentive equilibrium}. 

An incentive to a player $p$ is a function $\iota_p {:}  V^*V_p {\rightarrow}
\mathbb R_{\geq 0}$ from the set of histories to incentives.  
Incentives can be extended to infinite play $\pi = \seq{v_0, v_1, \ldots}$ in
the usual mean-payoff fashion:  
$\iota_p(\pi) \rmdef \liminf_{n\rightarrow \infty}
\frac{1}{n}\sum_{i=0}^{n-1}\iota_p(v_0 \ldots v_{n-1})$. 
The overall payoff $\rho_p(\pi)$ to a follower in run $\pi$ is the raw payoff
plus all incentives, $\rho_p(\pi) \rmdef r_p(\pi) + \iota_p(\pi)$, while the
overall payoff of the leader $\rho_l(\pi)$ is her raw payoff after deducting
all incentives, $\rho_l(\pi) \rmdef r_l(\pi) - \sum_{p \in
	P\smallsetminus\{l\}}\iota_p(\pi)$.  

We extend the notion of a strategy profile in the presence of incentives as a pair  
$(\osigma,\ogamma)$, where $\osigma$ is a strategy profile assigned by the
leader, in which the leader pays an incentive given by the incentive profile $\ogamma =
(\iota_p)_{p \in P\smallsetminus\{l\}}$. 
We write $\ogamma_p$ for the incentive for player $p \in P \setminus \set{l}$.  
We write $\ogamma_p(\osigma)$ for the incentive to player $p$ for the unique run
$\pi_\osigma$ under incentive profile $\ogamma$. 
In any incentive strategy profile $(\osigma,\ogamma)$, no player but
the leader may benefit from deviation. 
An optimal strategy profile among this class would form an \emph{incentive
	equilibrium}.
%

\begin{definition}[Incentive Equilibria]
	A strategy profile $(\osigma,\ogamma)$ is an incentive strategy
	profile, if no follower can improve his overall payoff from a unilateral
	deviation, i.e.,
	for all players $p \in P\smallsetminus\{l\}$ we have that 
	$r_p(\osigma) + \ogamma_p(\osigma) \geq
	r_p(\osigma_{p, \sigma'}) + \ogamma_p(\osigma_{p, \sigma'}) \text{   for all $\sigma'
		\in \Sigma_p$}$.
	An incentive profile $(\osigma,\ogamma)$ is an
	\emph{incentive equilibrium} if the leader's total payoff for this profile is maximal among all
	incentive strategy profiles. I.e., for all  
	$(\osigma',\ogamma')$ we have that 
	$ r_l(\osigma) - \sum_{p \in
		P\smallsetminus\{l\}}\ogamma_p(\osigma) \geq r_l(\osigma') - \sum_{p
		\in P\smallsetminus\{l\}}\ogamma'_p(\osigma')$.
	
	For a given $\varepsilon > 0$ we call an incentive strategy profile $(\osigma,
	\ogamma)$ an $\varepsilon$-incentive equilibrium if the leader's payoff is at
	most $\varepsilon$ worse than that of any other profile, i.e., for all 
	profiles $(\sigma',\ogamma')$ we have that 
	$r_l(\osigma) - \sum_{p \in
		P\smallsetminus\{l\}}\ogamma_p(\osigma) \geq r_l(\osigma') - \sum_{p
		\in P\smallsetminus\{l\}}\ogamma'_p(\osigma') - \varepsilon$.
\end{definition}

\paragraph{Incentive equilibria vs.\ leader equilibria.}
We call an incentive strategy profile a \emph{leader strategy profile} if
all incentives are constant $0$ functions, and a \emph{Nash strategy profile} if,
in addition, $\osigma$ is also a Nash equilibrium.
We write SP, ISP, LSP, and Nash SP for the set of strategy profiles,
incentive strategy profiles, leader strategy profiles, and Nash
strategy profiles. 
It is clear that  
$\text{ Nash SP} \subseteq \text{LSP} \subseteq \text{ISP} \subseteq \text{SP}$.
This observation, together with Example~\ref{fig:incentiveisbetter}yield the following result.
\begin{theorem}
	\label{theo:superior}
	Incentive equilibria do not provide smaller return than leader
	equilibria, and leader equilibria do not provide smaller return than
	Nash equilibria.
	Moreover, there are games for which the leader reward from three equilibria
	are different.  
\end{theorem}

\section{Existence and construction}
\label{sec:exist}
This section is dedicated to the existence and construction of incentive
equilibria. 
We first introduce a canonical class of incentive strategy
profiles---the \emph{perfectly-incentivised strategy profiles}
(PSPs)---that corresponds to the Stackelberg version of the  
classic subgame perfection.
Keep in mind that not all perfectly-incentivised strategy profiles (PSPs) are valid
incentive strategy profiles (ISPs). 
On the other hand, we show that every ISP has a corresponding PSP (which is
also an ISP) with the same leader reward.
Thanks to this result, in order to construct incentive equilibrium it suffices
to consider PSPs that are also ISPs. 

Further, we show that, for PSPs that are ISPs, it suffices to find a maximum in
a \emph{well behaved} class of strategy profiles: strategy profiles where every
edge has a limit share of the run---by showing that the supremum of general
strategies cannot be higher than the supremum of these well behaved ones. 
We then show how to construct well behaved PSPs that are ISPs based on a family
of constraint systems that depend on the occuring and recurring vertices on the
play. 
At the same time, we show that no general ISP that defines a play with this set
of occuring and recurrent vertices can have a higher value. 
The set of occuring and recurrent vertices can be guessed and the respective
constraint system can be build and solved in polynomial time, which also
provides inclusion of the related decision problem in NP. 


\subsection{Perfectly-incentivised strategy profiles}

We define a canonical form of an incentive equilibrium with this play that we
call \emph{perfectly-incentivised strategy profiles} (PSP).  
In a PSP, a deviator (a deviating follower) is punished, and the leader
incentivises all other followers to collude against the deviator. 
While the larger set of strategies and plays that define them (when compared to
Nash and leader equilibria) lead to a better value, this incentive scheme
leads to a higher stability: the games are subgame perfect relative to the
leader. 

\begin{definition} [Subgame Perfect]
	A strategy profile $(\osigma, \ogamma)$ is a subgame perfect incentive
	strategy profile, if every reachable subgame is also an incentive strategy profile.
\end{definition}

This term adjusts the classic notion of subgame perfect equilibria to our setting.
Subgame perfection refers to believable threats: broadly speaking, when a player
threatens to play an action that harms herself, then it may happen that the
other players do not believe this player and therefore deviate. 
In a subgame perfect Nash equilibrium, it is therefore required that the subgame
started on each history 
also forms a Nash equilibrium.
Note that the leader 
is allowed to benefit
from deviation in our setting.

The means to obtain subgame perfection after deviation is to make all players harm the most recent deviator. Thus, we essentially resort to a two-player game.
For a multi-player mean-payoff game $\mathcal{G}$, we define, for each follower
$p$, the two-player mean-payoff game (2MPG)  $\mathcal{G}_p$ where $p$ keeps his reward function, while
all other players have the same antagonistic reward $-r_p$.  
Two-player mean-payoff games are memoryless determined, such that every vertex
$v$ has a value, which we denote by $r_p(v)$.
This value clearly defines a minimal payoff of a follower: when he passes by a
vertex in a play, then he cannot expect an outcome below $r_p(v)$, as he would
otherwise deviate.

PSP strategy profiles are in the tradition of reward and punish strategy
profiles \cite{6940380}. In any 'reward and punish' strategy profile, the leader
facilitates the power of all remaining followers to punish a deviator.  
If a player $p$ chooses to deviate from the strategy profile at history $h$, the
game would turn into a two-player game, where all the other followers and the
leader forsake their own interests, and jointly try to `punish' $p$. 
That is, player $p$ may still try to maximise his reward and his objective
remains exactly the same, but the rewards of the rest of the players have changed
to negative of the reward of player $p$. 
As they form a coalition with the joint objective to harm $p$, this is an
ordinary two-player mean-payoff game that starts at the vertex $\mathsf{last}(h)$. 

For a strategy profile $\osigma$ and a history $h$, we call $h$ a deviating
history, if it is not a prefix of $\pi_\osigma$. 
We denote by $\mathsf{dev}(h,\osigma)$ the last player $p$, who has deviated from
his or her strategy $\osigma_p$ on a deviating history $h$. 

\begin{definition}
	[Perfectly-Incentivised Strategy Profile]
	A perfectly-incentivised strategy profile is defined as a strategy profile
	(PSP) $(\osigma,\ogamma)$ with the following properties. 
	For all prefixes $h$ and $h'$ of $\pi_\sigma$ and for all followers $p$, it
	holds that $\ogamma_p(h) = \ogamma_p(h')$. We also refer to this value by
	$\ogamma_p$. 
	For deviator histories $h'$, the incentive $\ogamma_{p}(h')$ is $0$ except for
	the following cases. 
	On every deviating history $h$ with deviating player $p =
	\mathsf{dev}(h,\osigma)$, the player $p'$ who owns the vertex
	$v=\mathsf{last}(h)$ follows the strategy from the 2MPG $\mathcal G_p$. 
	If, under this strategy, player $p'$ selects the successor $v'$ at a vertex $v$ in
	the 2MPG $\mathcal G_p$ (and thus $\osigma_{p'}(h)= v'$), $p'$ is a follower,
	\emph{and} $p' \neq p$, then player $p'$ receives an incentive, such that
	$r_{p'}(v,v') + \ogamma_{p'}(h\cdot v') = r_{\max} +1$. 
\end{definition} 

Note that, technically, the leader punishes herself in this definition. This is
only to keep definitions simple; she is allowed to have an incentive to deviate,
and the subgame perfection does not impose a criterion upon her. 
Note also that a PSP is not necesarily an incentive strategy profile, as it
does not guarantee anything about $\pi_\osigma$. 
The following theorem states the importance of PSPs in constructing incentive
equilibrium. 
\begin{theorem}
	\label{theo:ISP}
	Let $(\osigma,\ogamma)$ be an ISP that defines a play $\pi_{\osigma}$.
	Then we can define a  PSP $(\osigma, \ogamma)$, which is also an ISP, with the same
	reward that defines the same play. 
\end{theorem}
The proof of this theorem follows from Lemma~\ref{lem:enough} and
Lemma~\ref{lem:done}. 

\begin{lemma}
	\label{lem:enough}
	Let $(\osigma',\ogamma')$ be a strategy profile that defines a play
	$\pi_{\osigma'}$, which contains precisely the reachable vertices $Q$. 
	Let  $(\osigma',\ogamma')$ satisfy that, for all followers $p \in P \setminus \set{l}$ and
	all vertices $v \in Q \cap V_p$ owned by $p$ we have that 
	$\ogamma_p(\osigma') + r_p(\osigma') \geq r_p(v)$. 
	Then we can define a  PSP $(\osigma, \ogamma)$ with the same
	reward, which defines the same play. 
\end{lemma}

\begin{proof}
	We note that a PSP $(\osigma, \ogamma)$ is fully defined by the
	play $\pi_\osigma$ and the $\ogamma$ restricted to the prefixes of $\pi_\osigma$. 
	We now define the PSP $(\osigma, \ogamma)$ with the following property:
	$\pi_\osigma = \pi_{\osigma'}$, that is the play of the PSP equals the play
	defined by the ISP we started with. 
	For all followers $p$ and all prefixes $h$ of $\pi_\osigma$, we have $\ogamma_p(h)
	= \ogamma_p'(\osigma)$. 
	It is obvious that  $(\osigma',\ogamma')$ and  $(\osigma,
	\ogamma)$ yield the same reward for all followers and  
	the same reward for the leader. We now assume for contradiction that the
	resulting PSP is not an incentive strategy profile. 
	If this is the case, then a follower $p$ must benefit from deviation at some
	history $h$. 
	Let us start with the case that $h$ is a deviator history.
	In this case, the reward for $p$ upon not deviating is $r_{\max}+1$, while it is
	the outcome of some game upon deviation, which is clearly bounded by $r_{\max}$. 
	
	We now turn to the case that $h$ is not a deviator history, and therefore a
	prefix of $\pi_\osigma$. 
	Let $p$ be the owner of $v = \mathsf{last}(h)$.
	If $p$ is the leader, we have nothing to show.
	If $p$ is a follower and does not have an incentive to deviate in $(\osigma, \ogamma)$, we have nothing to show. 
	If $p$ is a follower and has an incentive to deviate in $(\osigma, \ogamma)$, we note that his payoff after deviation would be bounded from above by
	$r_p(v)$. 
	Thus, he does not have an incentive to deviate (contradiction). 
	\qed
\end{proof}

\begin{lemma}
	\label{lem:done}
	Let $(\osigma,\ogamma)$ be an ISP that defines a play
	$\pi_{\osigma}$, which contains precisely the vertices $Q$. 
	Then, for all followers $p \in P \setminus \set{l}$ and all vertices $v \in
	Q \cap V_p$ owned by $p$,  we have that 
	$\ogamma_p(\osigma') + r_p(\osigma') \geq r_p(v)$.
\end{lemma}

\begin{proof}
	Assume that this is not the case for a follower $p$ and a vertex $v \in Q$ owned
	by $p$. 
	Then $p$ would benefit upon deviating when visiting $v$.
	\qed
\end{proof}

\subsection{Existence and construction of incentive equilibria}
We say that a strategy profile $\osigma$ is \emph{well-behaved} if in the
resulting play $\pi_\osigma$, 
the frequency (ratio)  of occurrence of every edge of the game arena occurs has
a limit, i.e, each edge here occurs with a limit probability (the limes inferior
and superior of the share of its occurrence on $\pi_\osigma$ are equal). 
Such notion of well-behaved strategy profiles were also defined in~\cite{6940380} for the
case of leader equilibria.
We first show how to construct optimal ISPs among well behaved PSPs, and then
show that no ISPs give a better payoff for leader. 

\paragraph{Characterisation of a well-behaved PSPs.}
Let $\osigma$ is a well-behaved perfectly-incentivised strategy profile and let
$Q$ be the set of vertices visited in $\pi_\osigma$ and $S \subseteq Q$ be
the set of vertices that are visited infinitely often (note that $S$ is strongly
connected). 
Let $p_{(s,t)}$ be the limit ratio (frequency) of occurrence of an edge $(s,t) \in
E \cap S \times S$ in $\pi_\osigma$ and let $p_v$ be the for the limit ratio of
each vertex $v \in S$. 

Thanks to the proof of Lemma~\ref{lem:enough}, the following constraint system
(linear program) characterises the necessary and sufficient conditions for the
well-behaved perfectly-incentivised strategy profile $\osigma$ to be an ISP. 
\begin{enumerate}
	\item $p_v =  0$ if $v \in V \smallsetminus S$ and $p_v \geq 0$ if $v \in S$. 
	\item $p_e = 0$ if $e \in E \smallsetminus S \times S$ and $p_e \geq 0$ if $e \in E \cap S \times S$ 
	\item $\sum_{v \in V} p_v = 1$
	\item $p_s = \sum_{(s,t)\in E} p_{(s,t)}$ for all $s \in S$ and $p_t =
	\sum_{(s,t)\in E} p_{(s,t)}$ for all $t \in S$ 
	\item $\ogamma_p + \sum_{e\in E} p_e r_p(e) \geq \max_{v\in Q}(r_{p}(v))$ where
	$r_p(v)$ is the value at vertex $v$ in the  2MPG $\mathcal G_p$ characterising
	minimum payoff expected by player $p$.
\end{enumerate}
The constraints presented above are quite self-explanatory. 
Constraints~$1$ and~$2$ state that the limit ratio of occurrence of a vertex and
edge is positive only when it is visited infinitely often. 
Constraint~$3$ expresses that the sum of ratio of occurrence of vertices is equal
to~$1$, while constraint~$4$ expresses the fact the limit ratio of a vertex should
be equal to limit ratios of all incoming edges, and equal to limit ratio of all
outgoing edges from that vertex. 
The last constraint stems from the proof of Lemma~\ref{lem:enough} combined with
the observation that reward $r_p(\osigma)$ of a player $p$ in $\osigma$ is
simply $\sum_{e\in E}p_e r_p(e)$, that is, it is the weighted sum of the raw
rewards of the individual edges. 
Before we define the objective function, we state a simple corollary from
the proof of Lemma \ref{lem:enough}. 

\begin{corollary}
	\label{cor:wellBehavedSP}
	Every well behaved PSP that is an ISP satisfies these constraints, and every
	well behaved strategy profile $(\osigma, \ogamma)$, whose play
	$\pi_\osigma$ satisfies these constraints, defines a PSP, which is then an ISP. 
\end{corollary}
Note that the resulting PSP is an ISP even if  $(\osigma, \ogamma)$
is not. This is because the satisfaction of the constraints are enough for the
final contradiction in the proof of Lemma \ref{lem:enough}. 

\paragraph{Construction of incentive equilibria.}
The objective of the leader is obviously to maximise $r_l(\osigma) -
\sum_{p \in P \smallsetminus \{l\}} \ogamma_p =  \sum_{e\in E}p_e r_l(e) - \sum_{p \in
	P \smallsetminus \{l\}} \ogamma_p$. 
Once we have this linear programming problem, it is simple to determine a
solution in polynomial time \cite{Karmarkar/84/Karmarkar,Khachian/79/Elliptic}. 
We first observe that it is standard to construct a play defining a PSP from a
solution. 
(A description is given in the appendix.)

A key observation is that, if the linear program detailed above for sets $Q$ of
reachable vertices and $S$ of vertices visited infinitely often has a solution, then
there is a well behaved reward and punish strategy profile that meets this
solution.  

\begin{theorem}
	\label{theo:wbCSSPs}
	Non-well behaved PSPs that are also ISPs cannot provide better rewards for the
	leader than those from well behaved PSPs that are also
	ISPs. 
\end{theorem}

\begin{proof}
	Corollary~\ref{cor:wellBehavedSP} shows that there exists a well defined
	constraint system obeyed by all well behaved PSPs that are also ISPs with a set
	$Q$ of reachable vertices and a set $S$ of recurrent vertices. 
	
	Let us assume for contradiction that there is a reward and punish strategy
	profile $(\osigma, \ogamma)$ that defines a play $\pi_\osigma$ with
	the same sets $Q$ and $S$ of reachable and recurrent vertices, respectively, that
	provides a strictly better reward $r_l(\osigma) - \sum_{p \in P
		\smallsetminus \{l\}} \ogamma_p$, which exceeds the maximal reward obtained by the
	leader in well behaved PSPs that are also ISPs by some $\varepsilon >0$. 
	
	We now construct a well behaved PSPs that are also ISPs and that also provides a
	better return. 
	First, we take a $\ogamma'$ with $\ogamma_p = \ogamma_p'$ for all followers $p$. This
	allows us to focus on the raw rewards only. 
	
	Let $k$ be some position in $\pi_\osigma$ such that, for all $i \geq k$, only
	positions in the infinity set $S$ of $\pi_\osigma$ occur. 
	Let $\pi$ be the tail $v_k v_{k+1} v_{k+2} \ldots$ of $\pi_\osigma$ that starts
	in position $k$. 
	Obviously $r_p(\pi)=r_p(\osigma)$ holds for all players $p \in P$.

	We observe that, for all $\delta > 0$, there is an $l \in N$ such that, for all
	$m \geq l$, $\frac{1}{m}\sum_{i=0}^{m-1} r_p\big((v_i,v_{i+1})\big) > r_p(\pi) -
	\delta$ holds for all $p \in P$, as otherwise the limes inferior property would
	be violated.

	We now fix, for all $a \in \mathbb N$, a sequence $\pi_a = v_k v_{k+1} v_{k+2}
	\ldots v_{k+m_a}$, such that $v_{k+m_a+1}= v_k$ and
	$\frac{1}{m}\sum_{i=0}^{m_a-1} r_p\big((v_i,v_{i+1})\big) > r_p(\pi) -
	\frac{1}{a}$ holds for all $p \in P$.

	Let $\pi_0 = v_0 v_1 \ldots v_{k-1}$. We now select $\pi'=\pi_0 {\pi_1}^{b_1}
	{\pi_2}^{b_2} {\pi_3}^{b_3} \ldots$, where the $b_i$ are natural numbers big
	enough to guarantee that $\frac{b_i \cdot |\pi_i|}{|\pi_{i+1}| + |\pi_0| +
		\sum_{j = 1}^i b_j \cdot |\pi_j| } \geq 1 - \frac{1}{i}$ holds.

	Letting $b_i$ grow this fast ensures that the payoff, which is at least
	$r_p(\pi) - \frac{1}{i}$ for all players $p \in P$, dominates till the end of
	the first iteration%
	\footnote{Including the first iteration of $\pi_{i+1}$ is a technical necessity,
		as a complete iteration of $\pi_{i+i}$ provides better guarantees, but without
		the inclusion of this guarantee, the $\pi_j$'s might grow too fast, preventing
		the existence of a limes.} 
	of $|\pi_{i+1}|$.

	The resulting play belongs to a well behaved (as the limit exists) strategy
	profile, and can thus be obtained by a well behaved PSP by Corollary \ref{cor:wellBehavedSP}. 
	It thus provides a solution to the linear program from above, which contradicts
	our assumption. 
	\qed
\end{proof}

Consequently, it suffices to guess the optimal sets $Q$ of vertices that occur
and $S$ of vertices that occur infinitely often to obtain a constraint system
that describes an incentive equilibrium, which is well behaved and a
PSP---and therefore subgame perfect. 

\begin{corollary}
	The decision problems `is there a (subgame perfect) incentive equilibrium with
	leader reward $\geq r$' is in NP, and the answer to these two questions is the
	same. 
\end{corollary}

Note that, if we have a fixed number of players, the number of possible
constraint systems is polynomial. Like in \cite{6940380}, there are only
polynomially many (for $n$ vertices and $k$ followers $O(n^k)$ many) second parts
(the constraints on the follower rewards) of the constraint systems. For them,
it suffices to consider the most liberal sets $Q$ (which is unique) and $S$ (the
SCCs in the game restricted to $Q$, at most $n$). 
For a fixed number of players, finding incentive equilibiria is therefore in the
same class as solving 2MPGs.
By adapting (Section~\ref{sec:complexity}) the NP hardness proof for leader
equilibrium in mean-payoff games from \cite{6940380} we get the following results. 

\begin{theorem}
	\label{NPhard1}
	The problem of deciding whether an incentive equilibrium $\sigma$ with
	reward $r_l(\sigma) \geq 1-1/n$ of the leader exists in games with rewards in
	$\{0,1\}$, is NP-complete.
\end{theorem}

\subsection{Secure $\varepsilon$ incentive strategy profiles}
\label{sec:security}

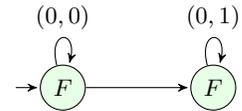
\begin{wrapfigure}[7]{r}{3.2cm} 
	\vspace{-8mm}
	\begin{tikzpicture}[->,>=stealth',shorten >=1pt]
	\tikzstyle{vertex}=[circle,fill=black!10,minimum size=17pt,inner sep=0pt,font=\sffamily\small\bfseries]
	\tikzstyle{rvertex}=[circle,fill=red!10,minimum size=17pt,inner sep=0pt,font=\sffamily\small\bfseries]
	\tikzstyle{gvertex}=[circle,fill=green!10,minimum size=17pt,inner sep=0pt,font=\sffamily\small\bfseries]
	\node (6) at (0.25,1) {};
	\node (1) at (1,1) [gvertex,draw]{$F$} ; 
	\node (2) at (3,1) [gvertex,draw] {$F$};
	
	\path[every node/.style={font=\sffamily\small}]
	(6) edge [right] node[] {} (1)
	(1) edge [right] node[] {} (2)       
	edge [loop above] node[] {$(0,0)$} (1)
	(2) edge [loop above] node[] {$(0,1)$} (2);
	
	\end{tikzpicture}
	\caption{Secure equilibria.}
	\label{fig:secure}
\end{wrapfigure}

We take a short detour to another class of equilibria that make a solution
stable: secure equilibria. Secure equilibria  
\cite{Chatterjee:2006:GSE:1226608.1226612} have been defined as Nash equilibria
with the additional property that each player would, upon unilateral deviation,
either lose strictly, or no other player would lose. Naturally, we have to
adjust this definition appropriately. 
We say that a strategy profile $(\osigma, \ogamma)$ is a secure incentive
strategy profile, if, upon unilateral deviation, every follower either receives a
strictly lower reward, or an equal reward. In the latter case, all other
players have to receive at least the same reward as before. 
We show that we can obtain subgame perfect secure $\varepsilon$ incentive
equilibria (i.e., every subgame is a secure incentive strategy profile) by
simply increasing 
the individual incentives from the strategy we have constructed by
$\frac{\varepsilon}{|P|}$.
The payoff between secure $\varepsilon$ incentive equilibria and general
incentive equilibria is therefore arbitrarily small. 
This is in contrast to leader and Nash equilibria, where security can come to a
high cost.

In the simple example shown in Figure \ref{fig:secure} 
(rewards are shown in the
order (follower, leader)), where the left vertex is owned by the follower, the
leader can incentivise the follower to move to the right vertex by an
arbitrarily small incentive $\varepsilon$, resulting in a secure incentive
strategy profile and payoffs of $1-\varepsilon$ and $\varepsilon$ for the leader
and her follower, respectively. 
A secure leader (and Nash) equlibrium would require the follower to stay forever
in the left vertex, resulting in a payoff of $0$ for the leader and her follower
alike. 
This is in contrast to `normal' leader (or Nash) equilibria, which would allow
for the follower moving the token to the right, resulting in a payoff of $1$ and
$0$ for the leader and her follower, respectively. 

\section{Experimental results}
\label{sec:tooldetails}
We have implemented a tool~\cite{mmpgsolver} in C++ to evaluate the performance of the proposed
algorithms for multi-player mean-payoff games (MMPG) for a small number
of players. We implemented an algorithm from \cite{Schewe/08/improvement} to find mean values at the vertices. 
We then infer and solve a number of constraint systems. We describe our main algorithm here.

\subsection{Algorithm specific details}
We first evaluate MMPGs using reduction to solving underlying 2MPGs. 
We then infer and solve a number of linear programming problems to find a solution. 
For few number of players, the number of different solutions to these games is usually small, and,
consequently, the number of linear programming problems to solve is small,
too. 
In order to find the individual mean partition, we use an algorithm from \cite{Schewe/08/improvement},
that finds $0$-mean paritions, and expand it quantitatively to find the value of 2MPGs. 
We recall that for 2MPG both players have optimal memoryless strategies.
Under such strategies, the game will follow a `lasso path' from every starting
vertex: a finite (and possibly empty) path, followed by a cycle, which is
repeated infintiely many times. The value of a game position is defined by the
average of the edge weights on this cycle. 

In our context, the edge weights are either $0$ or $1$.
The values of the vertices are therefore fractions $\frac{a}{l}$ with $0 \leq a
\leq l \leq n$, where $l$ is the length of the cycle, and $a$ is the number of
`accepting' events in the DBA that refers to the objective of the respective
player, i.e., the edges with value $1$, occurring on this cycle.  


An $\alpha$-mean partition of a 2MPG is the subset of vertices, for which the
return is $\geq \alpha$.
Conceptually, to find the $\frac{a}{l}$-mean partition, one would simply
subtract $\frac{a}{l}$ from the weight of every edge and look for the $0$-mean
partition. 
However, to stay with integers, it is better to use integer values on the edges,
e.g., by replacing the $0$s by $-a$, and the $1$s by $l-a$. 
For games with $n$ vertices, there are only $O(n^2)$ values for the fraction $\frac{a}{l}$ to consider, 
as optimal memoryless strategies always lead to lasso paths and only the cycle at the end of the lasso 
determines the values for $a$ and $l$, where $0<a<l\leq n$.

We start by narrowing down the set of values by classifying the mean partition in a logarithmic search.
After determining the $\frac{1}{2}$ mean partition, we know which values are $<0.5$ and $\geq 0.5$, respectively.
The two parts of the game can then be analysed further, determining the $\frac{1}{4}$ and $\frac{3}{4}$ mean partition, respectively.
After $s$ such partitionings, all values in a partition of the game are either known to be in an $[k \cdot 2^{-s},(k+1) \cdot 2^{-s}[$ interval for some $k<2^s -1$, or in the interval $[1-2^{-s},1]$.
We stop to bisect when the size $p$ of a partition is at most $2^s$.
In this case, the respective interval has $f \leq p$ fractions with a denominator $\leq p$.
We determine them, store them in a balanced tree, and use it to determine the correct value of all vertices of the partition in $\lceil \log_2 f \rceil$ steps.

\emph{Solving multiplayer mean-payoff games.}
\begin{enumerate}
	\item Initially, we start with a $\frac{1}{2}$ mean partition. For the $\geq$ part of the game, we continue with a $\frac{3}{4}$ mean partition, and so forth.
	After $s = \lceil \log_2 n \rceil$, we have narrowed the area down to an interval of length $2^{-s}$, and we know that the value lies within this interval.
	\item For each denominator, there is at most one numerator in this interval%
	\footnote{Exception: $n = 2^s$ But then we can simply look for the $1$ mean partition and are done.}.
	Thus, going through all possible denominators, we can then sort the resulting fractions in a balanced tree. It suffices to take those, which are relative prime. (The value is otherwise already in the balanced tree.)
	\item We then use the values stored in the balanced tree to find mean partitions, starting with the root. At most height-of-the-tree many further iterations are needed ($O(\log n)$ many).
\end{enumerate}

The number of different values of nodes in a 2MPG is usually small, and certainly it would be much smaller than the number of vertices in the game.
Consequently, the number of constraint systems is also small for a small number of players.

We use this algorithm to evaluate a number of randomly created three player MPGs, where the player take turns. 
We consider three players -- player 1, player 2 and a leader and two different evaluations on the same game graph.
We first see how each player fares when they try to maximise their return against a coalition of all other players, including the leader.
In the first evaluation, leader forms a coalition with player~$1$ (minimiser) against player~$2$ (maximiser) on the payoffs defined for player~$2$. 
We find the different possible mean values at the nodes in this evaluation, using the algorithm from above.
In the second evaluation, leader forms a coalition with player~$2$ (minimiser) against player~$1$ (maximiser) on the payoffs defined for player~$1$. 
We also note the different possible mean values at the nodes in this evaluation, using again the algorithm from above.

The resultant two-player games provide the constraints for the linear programming problems. 
These different values form the different thresholds that we have to consider. 
We now consider all possible combinations of these different threshold values for the followers and determine the vertices that comply with them. 

For each set of vertices, we then do the following:

a) We first recursively remove the nodes that have no successor;
b) We then remove the nodes that are not reachable from the initial state, i.e., we determine the set of reachable vertices;
c) We determine the strongly connected set of components (SCCs); and
d) For the SCCs formed from above, we build and solve the respective linear program.

\emph{Constraints on SCCs.}
To construct the linear programs over the SCCs formed from above, we have side-constraints on the edge-ratio and vertex-ratio (these constraints
are to comply with the limit behaviour of nodes and edges) and we have constraint over the reward of player. 
The constraints over ratios of edge and vertices and over reward of the players can be seen in detail in \cite{6940380}.
Additionally, to construct an incentive equilibrium, we have a constraint over incentives here.
We first have constraints on the nodes and edges that form part of the strongly connected component $S$:
a)  ratio of vertices and edges that are not part of $S$ is $0$,
b)  ratio of vertices and edges that are in $S$ is $\geq 0$,
c)  sum of the ratio of vertices is $1$

The second part of constraint system is constraint over rewards:

\begin{itemize}
	\item for every player $p$ other than the leader, that own some vertex in $S$, we have constraint over her reward 
	$$\iota_p + \sum_{e\in E}p_e r_p(e) \geq \max_{v\in Q}(r_{p}(v))$$
	where $p_e$ is the ratio by which edge $e$ is taken, $r_p(e)$ is the edge weight for player $p$ at edge $e$, $r_{p}(v)$ is the mean value of game for player $p$ in set $S$ and $\iota_p$ is an incentive given to player $p$ and $\iota_p \geq 0$
	\item objective of the constraint system is to maximise leader's reward, i.e., maximising reward at leader nodes in $S$. This gives us objective function: maximise $\sum_{e\in E}p_e r_l(e) - \sum_{p\in P}\iota_p$
\end{itemize}
For a set $Q$, we may have number of $SCCs$ and there is a constraint system for every such component. 
In this case, we would take the one that maximises leader's reward.
\begin{figure}[htb]
	\centering
	\includegraphics[width=.7\linewidth, scale=0.27]{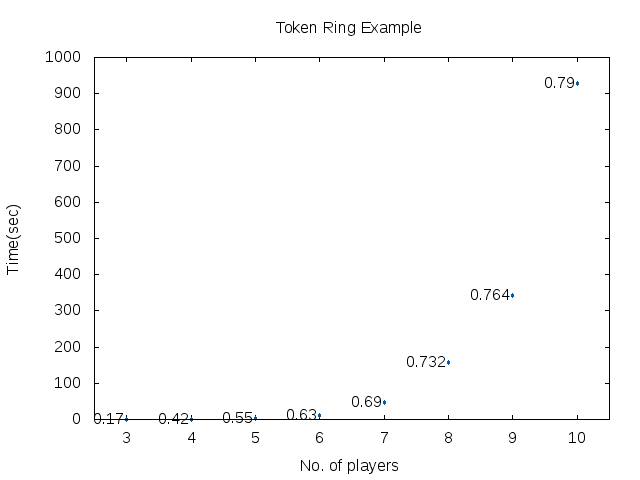}
	\caption{
		The figure shows results for a generalisation of example $2$ for multiple
		players with $n$ nodes in the inner cycle and $n-1$ nodes in outer cycles
		where $n$ is the number of players.
	}  
	\label{fig:cegar1}
\end{figure}

\subsection{Experimental results} 
Experiments indicate that our implementation of the algorithm can solve
examples of size $100$ nodes and $10$ players within 30 minutes. 
The algorithm is, of course, much faster for the games with two or three players. 
Figures~\ref{fig:cegar1} and~\ref{fig:cegar2} show the experimental results
for the following two problem classes.

\begin{figure}[t]
	\centering
	\includegraphics[width=.5\linewidth, scale=0.2]{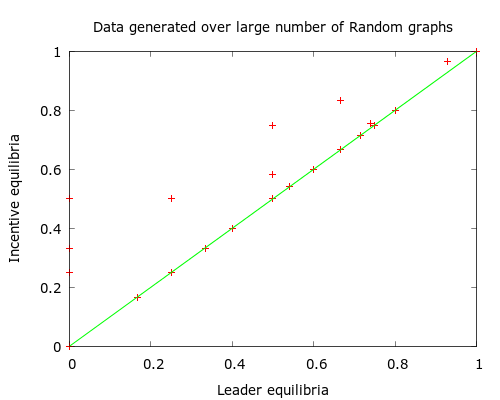}\includegraphics[width=.5\linewidth, scale=0.2]{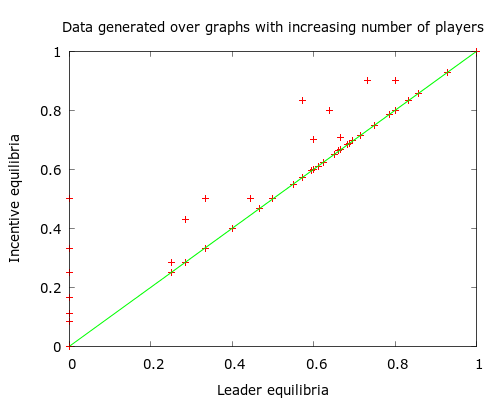}
	\caption{
		The left figure shows the results for randomly generated MMPGs with
		3 players, while the right one is for randomly generated MMPGs with 3 to 10
		players. 
	}
	\label{fig:cegar2}
\end{figure}


\begin{itemize}
	\item 
	Recall the example from Figure~\ref{fig:intro-token}. 
	We generalise this example for token ring graph parameterised by 2 variables, n
	and d. It has 'n' nodes on the inner cycle, each of which correspond to 'n' 
	different players and each of these 'n' nodes is also present on another cycle
	of length 'd'. The weights are set such that, all players except the leader get
	'1/n' if they chose the inner ring and get '1/d' if they chose their
	respective outer ring. The leader reward is '1' in the inner ring and
	'1/d' in all the other rings.
	The data supports the pen-and-paper analysis that incentives are useful iff $n > d >  n(n-1)/(2n-1)$ holds.
	Figure~\ref{fig:cegar1} shows the leader reward for this example and the running time of our tool to compute it.
	
	\item
	Figure~\ref{fig:cegar2} (left plot) shows the difference between
	incentive equilibrium and leader equilibrium for randomly generated $3$ player
	MMPGs, while the right plot shows similar results on random graphs, where
	the number of players range from $3$ to $10$. 
\end{itemize}

The evaluation results 
confirm that the leader reward increases significantly in incentive equilibria when compared
to leader equilibria.

\section{Discussion}
\label{sec:dis}
The main contribution of this paper is the introduction of incentive equilibria
in multi-player mean-payoff games and the implementation of our techniques in a
tool. We study how a rational leader might improve over her outcome by paying
small incentives to her followers. At first, it may not seem to be a rational
move of the leader, but close insight would show how a leader might improve her
reward in this way.
The incentive equilibria are seen as an extension to leader equilibria, where a
rational leader, by giving an incentive to every other player in the game, can
derive an optimal strategy profile. We believe that these techniques are helpful
for the leader when maximising the return for a single player and would also be
instrumental in defining stable rules and optimising various outcomes. 
The evaluation results from Section \ref{sec:tooldetails} show that the results
are significantly better for the leader in an incentive equilibrium as compared
to her return in a leader equilibrium.
  
\appendix
\newpage
\section*{Appendix}
\section{From $Q$, $S$, and a solution to the linear programs to a well behaved reward and punish strategy profile}
We start with the simple case that the vertices and edges with non-$0$ ratio are strongly connected.

We design $\pi_\sigma$ as follows.
We first go from the initial vertex $v_0$ through states in $Q$ to some state in $S$.
(Note that this initial path has no bearing on the lower limit that defines the payoff of the individual players.)

Once we have reached $S$, we intuitively keep a list for each vertex in $S$.
In this list, we keep the number of times each outgoing edge with non-$0$ ratio has been taken.
We also apply an arbitrary (but fixed) order on the outgoing edges.
Each time we are in this vertex, we choose the first edge (according to this order) that has been taken less often 
 (from this vertex) than $\frac{p_e}{p_v}$, the ratio $p_e$ of the edge divided by the ratio $p_v$ of this vertex, suggests. If no such edge exists, we take the first edge.


The result is obviously a well behaved strategy profile and the first part of the constraint system is clearly satisfied.
It therefore suffices to convince ourselves that the second part is satisfied as well.

Now assume for contradiction that this is not the case.
Let $q_v$ and $q_e$ be the real ratio of the vertices and edges, respectively.
Note that our simple rule for the selection of vertices implies that $\frac{p_e}{p_v}$ is correct for all edges $e = (v,v') \in E \cap S \times S$.
Then there must be a vertex $v \in S$, which has the highest factor $\frac{q_v}{p_v}$.
As it is the highest factor, none of its predecessors in $E \cap S \times S$ can have a higher ratio; consequently, they must have the same ratio.
By a simple inductive argument, this expands to the complete strongly connected set of non-$0$ vertices.
As $\sum_{v \in S} p_v = 1 = \sum_{v \in S} q_v$ holds, this implies $p_v = q_v$ for all $v \in S$.

To extend this argument to the general case, we first observe that the non-$0$ vertices and edges form islands of (maximal) SC parts $C_1$, through $C_k$.
We use this observation to compose a play as follows.

We start with an initial part, a transfer from $v_0$ to $C_1$ as in the simple case.
We then continue by playing a $C_1^1$ part, a transfer, a $C_2^1$ part, a transfer, $\ldots$, a $C_k^1$ part, transfer $C_1^2$, and so forth. 
To achieve a well behaved strategy profile we do the following.
\begin{enumerate}
\item We fix the ratio  $\sum_i C_1^i : \sum_i C_2^i : \ldots : \sum_i C_k^i$ according to the the sum of the $p_v$ for vertices $v$ in the respective component. This ratio never changes, and it is given by natural numbers $c_1, c_2, \ldots, c_k$, such that $c_1 : c_2 : \ldots : c_k$ satisfies this ratio.
 
\item  We let $C_j^i$ grow slowly with $i$. We can, for example, use $i \cdot c_j$. 
 
 Note that the transfer part has constant length, bounded by $|S|$.
 Thus the limit ratio of transfer is $0$.

\item We let the transfer to $C_{j+1}^i$ go to the vertex, in which $C_j^i$ was left. Note that the transfer may contain vertices of various components, but as the overall ratio of the transport is $0$, this does not affect the limit probability.
 
Thus, we can use the controller from the simple case of one SCC for the sequence $C_i^1, C_i^2, C_i^3 \ldots$, which only focuses on the relevant part of the $i^{th}$ component.
\end{enumerate}

In effect, we have simple controllers for the individual components, and a single counting controller that manages the transfer between the components.

It is easy to see that the resulting controller inherits the right ratios from the simple individual controllers.
Together with Corollary \ref{cor:wellBehavedSP} we get:

\paragraph{Observation.}
If the linear program from above for sets $Q$ of reachable states and $S$ of states visited infinitely often has a solution, then there is a well behaved reward and punish strategy profile that meets this solution.
\qed

\section{Implementation related details}
For an efficient implementation we restricted our rewards to $0$ and $1$. 
This class of MMPG is sufficient to solve quantitative Buchi game problems where
the goal of each player is to maximise the limit share of time spend in
accepting states.  
A practical example of such a situation is shown below. 
\subsection{Paradigmatic examples}
Consider a client and server application where client is responsible for making
requests and server is responsible for granting or scheduling access to the
resource being requested. The client and server programs are depicted as a
deterministic buchi automata (DBA) and their objectives remain to maximise the
limit-average share of the time they would spent in an accepting state. We refer
to DBA $\mathcal A$ from Figure \ref{fig:client} and DBA $\mathcal B$ from
Figure \ref{fig:server} for the objectives of client and server in respective
order. 
While $\mathcal A$ would try to maximise the time spent in accepting state and
thereby increasing her limit share of reward, $\mathcal B$ would try to maximise
overall utilisation by trying to optimise return for both $\mathcal A$ and
$\mathcal B$. In this example, if $\mathcal A$ is in an accepting state or
utilising the critical resource, $\mathcal A$ would receive a utility of
1. While, if $\mathcal B$ is in an acceptance state or utilising the critical
resource, $\mathcal B$ would receive a utility of 10. If both $\mathcal A$ and
$\mathcal B$ alternatively take turns between accessing the critical resource,
overall return for $\mathcal A$ is $0.5$ while overall return for $\mathcal B$
is $5$.  
Contrary to this, a rational server would look at means of improving its
utilisation. 
For this, $\mathcal B$ may try to maximise the time spent in the acceptance
state so as to increase the overall resource utilisation.  
For example, $\mathcal B$ may take turn $\frac{3}{4}$ of the time to be spent in
critical section and would allow $\mathcal A$ to access the critical section for
$\frac{1}{4}$ of the total time.  
$\mathcal B$ may further give a small incentive amount of $\frac{1}{4}$ to $\mathcal A$
so as overall return for $\mathcal A$ is same as earlier, i.e., $\frac{1}{2}$ while
$\mathcal B$ now gains from $5$ to $7.25$. This way, return for $\mathcal B$
would be increased to $7.25$. 

\begin{figure}[!htb]  \centering
	\begin{minipage}{.5\textwidth}
		\centering
\begin{tikzpicture}[->,>=stealth',shorten >=1pt]
  \tikzstyle{vertex}=[circle,fill=black!10,minimum size=17pt,inner sep=0pt,font=\sffamily\small\bfseries]
\node (4) at (-0.25,0.5) {};
\node (1) at (0.5,0.5) [vertex,draw]{} ; 
\node [state,accepting](2) at (3.5,0.5) [vertex,draw]{} ;
\node (3) at (3,-1) [vertex,draw]{} ; 
\node (5) at (1,-1) [vertex,draw]{};

\path[every node/.style={font=\sffamily\small}]
  (4) edge [right] node[] {} (1)
  (1) edge [right] node[above] {$r'$} (2) 
	  edge [loop above] node[] {$r$} (1) 
  (2) edge [bend left]  node[below] {$r$} (1)
      edge [bend left]  node[right] {$s$} (3)
      edge [bend right, dashed] node[right] {$s'$} (3)
      edge [loop above] node[] {$g$} (2)
  (3) edge [bend left] node[below] {$s$} (5)  
      edge [below,dashed] node[below] {$s'$} (5)
  (5) edge [bend left] node[left] {$r$} (1);

\end{tikzpicture} 
\caption{\label{fig:client}}
\end{minipage}%
\begin{minipage}{0.5\textwidth}
	\centering
\begin{tikzpicture}[->,>=stealth',shorten >=1pt]
  \tikzstyle{vertex}=[circle,fill=black!10,minimum size=17pt,inner sep=0pt,font=\sffamily\small\bfseries]
\node (4) at (0.25,1) {};
\node (1) at (1,1) [vertex,draw]{} ; 
\node [state,accepting] (2) at (4,1) [vertex,draw] {};

\path[every node/.style={font=\sffamily\small}]
  (4) edge [right] node[] {} (1)
  (1) edge [right] node [above] {$r'$} (2) 
	  edge [loop above] node[]  {$s$} (1)      
	(2) edge [bend left] node [below] {$r$} (1)
      edge [loop above] node[] {$g$} (2);
            
\end{tikzpicture}
\caption{\label{fig:server}}
\end{minipage}
\end{figure}

In some state of this model, $\mathcal A$ has choices between two
edge-transitions to the same state. 
For example, in Figure \ref{fig:client}, $\mathcal A$ has the choice to take two
different edge-transitions to stay in critical section. 
On one edge transition, $s$, the edge-weight is $(0,1,-1)$ and there is another
edge transition $s'$ that has edge weight $(-1,3,-2)$.  
Here, pay-offs are in the order $\mathcal A$, $\mathcal B$ and a passive
player. 
If $\mathcal A$ alternates between taking the two transitions, $\mathcal A$
would receive a payoff of $-0.5$ and $\mathcal B$ gets $2$ here.  
In order to gain maximum benefit, $\mathcal B$ may assign a strategy to
$\mathcal A$ where instead of doing alternations between two transitions,
$\mathcal A$ would take edge transition $s$ with probability $\frac{1}{4}$ and
$s'$ with probability $\frac{3}{4}$.  
To incentivise $\mathcal A$ for this strategy, $\mathcal B$ gives an incentive
of $0.25$ to $\mathcal A$ and thus $\mathcal A$ would receive an overall return
of $-0.5$ and payoff for $\mathcal B$ is increased from $2$ to $2.25$. 

\section{Complexity}
\label{sec:complexity}
\begin{theorem}
  \label{NPhard1}
  The problem of deciding whether an incentive equilibrium $\sigma$ with
  reward $r_l(\sigma) \geq 1-1/n$ of the leader exists in games with rewards in
  $\{0,1\}$, is NP-complete.
\end{theorem}
\begin{proof}[of Theorem~\ref{NPhard1}]
The proof closely relates to the NP hardness proof for leader equilibrium in mean-payoff games from \cite{6940380}.
We consider reduction of the 3SAT satisfiability formula over $n$ atomic propositions with $m$ conjuncts to solve a MMPG, in order to establish NP Hardness. We assume the game graph has $2n+1$ players and $5m+4n+2$ vertices with payoffs 0 and 1 only. 
For reduction, we consider an example of a 3SAT formula $C_1 \wedge C_2 \wedge C_3$ with $C_1 = p \vee q \vee \neg r$, $C_2 = p \vee \neg q \vee \neg r$, and $C_3 = \neg p \vee q \vee r$. We have $2n$ players for the $2n$ literals that corresponds to the 
$n$ variables and there is one leader player. The game consists of three phases -- an initial assignment phase in which leader would make either truth or false assignment to the literal players. Then, we have a validation phase, in which leader intuitively tries to validate the 3SAT formula as per the assignment done in assignment phase. In the third evaluation phase, payoffs are rewarded to every player including leader. Assignment phase would consist of $2n$ literal vertices and $m$ leader vertices as leader chooses
literals corresponding to $n$ variables for the formula assignment. In Validation phase, there are 3 vertices in every conjunct of the 3SAT formula, i.e., $3m$ vertices. In the evaluation phase, we again have $2n$ literal vertices and a leader vertex where game goes round in a cycle of length $n$. At every point in this cycle, payoffs are given to the players and the leader.
There is additionally one sink vertex in the game graph, that has only one outgoing edge to itself. The sink vertex has a payoff of 1 for all literal players but a payoff of 0 for the leader. 
Game is terminated at sink vertex.
It is in the evaluation phase that by choosing the payoff for the players, it can be decided whether there exists an incentive equillibrium in the game with payoff of $1$ for the leader.

Starting in the assignment phase, for each of the $m$ conjuncts, there are $m$ leader vertices. For each conjunct, leader would select a literal vertex for each variable. Like, for a variable $'Z'$, leader would either choose an assignment $'z'$ or $'\neg z'$. 
The selected literal vertex then has two choices -- to continue the game by going to the next leader vertex or to termintae the game by choosing to go to sink vertex. 
The sink vertex has only one self loop that has only one outgoing edge with payoff of $1$ for all literal players and a payoff of $0$ for the leader. 
If literal vertex chooses to go to next leader vertex, leader would go with further assignments in the asignment phase. 

In the validation phase, leader intuitively tries to validate whether the 3SAT formula is satisfiable or not according to the chosen assignments in the assignment phase. For each conjunct and for each variable $'Z'$, leader either goes to $'z'$ where literal 
$'\neg z'$ receives a payoff of 0 and every other player and leader would receive a payoff of 1. Here, also, at every literal vertex, player may opt to continue the game by going to the next leader vertex or to terminate the game by going to the sink vertex.
In our example formula, the formula is satisfiable if leader in the assignment phase would select literals $p$, $\neg q$ and $r$. The incentive equilibrium would be a path $\big(p,\neg q,r)^\omega$. 
If the formula is not satisfiable, any run might have to path by both $'z'$ and $'\neg z'$ for a conjunct and a literal player at any point who receives a payoff of 0 might terminate the game by going to sink vertex -- that results in leader receiving a payoff 
of 0. Thus, for the strategy profiles that end up in the sink vertex and for the unsatisfiable formulae, an incentive equilibrium would have a payoff of 0 for the leader. In the validation phase, if a literal player deviates to sink vertex, all other player receive a payoff of 1. Leader, therefore, incentivise all remaining players to form a coalition and act against the deviating player. Leader can promise to pay an incentive of $1/n$ to every other player in the game.

While, if the formula is satisfiable, game futher goes to the evaluation phase, where nodes are owned by the leader.
Here, for a variable $'Z'$, leader either moves to $'z'$, where a payoff of $1$ is given to every player but 0 is given to $'\neg z'$or leader goes to $'\neg z'$ where a payoff of 1 is given to $'\neg z'$ and all other players, while a payoff of 0 is given to $'z'$.
Additionally, an incentive of $1/n$ is given to all other players. The leader's payoff in any incentive equilibrium, therefore,
equals to $1-1/n$, that is better than the payoff of 0 at the sink vertex for the
leader.  
The proof is now complete.
\end{proof}

\section{Secure $\varepsilon$ incentive strategy profiles}
Based upon the discussion from Section~\ref{sec:security}, we state the following theorem.

\begin{theorem}
	\label{theo:spe}
	We can obtain a secure subgame perfect $\varepsilon$ incentive equilibrium
	$(\osigma, \ogamma)$. 
	\end{theorem}
	
	\begin{proof}
		Using Theorem \ref{theo:wbCSSPs}, we can produce a well behaved incentive
		equilibrium, which is also a PSP. 
		Re-visiting the proof of Lemma \ref{lem:enough}, this PSP satisfies the
		requirements of a secure incentive equilibrium in every subgame that starts in a
		deviating history. 
		
		For non-deviating histories, however, we have that no follower benefits from
		deviation, but will normally lack the security property  
		(e.g., in the example from Figure~\ref{fig:incentiveisbetter}).
		We now produce a new PSP $(\osigma, \ogamma')$, such that $\ogamma'$
		is obtained from $\ogamma$ by selecting $\ogamma_p' = \ogamma_p +
		\frac{\varepsilon}{|P|}$ for all followers. 
		The subgames that start in a deviating history are not affected by this change,
		such that the resulting PSP also satisfies the requirements of a secure
		incentive equilibrium from these positions. 
		For non-deviating histories, however, we have now increased the value for
		following slightly, such that the pre-requisite of secure equilibria is
		satisfied here, too. (The deviating follower would strictly decrease his
		reward.) 
	\end{proof}

\end{document}